\documentclass[prd,aps,amsfonts,notitlepage,longbibliography,twocolumn]{revtex4-1} 

\usepackage{graphicx}
\usepackage{tabularx}
\usepackage{braket}
\usepackage{amssymb}
\usepackage{amsmath}
\usepackage{bm}
\usepackage[driverfallback=dvipdfm]{hyperref}
\hypersetup{colorlinks=true,breaklinks,urlcolor=blue,linkcolor=blue,citecolor=blue}

\newcommand{\revise}[1]{#1}
\newcommand{\red}[1]{#1}

\newcommand{\comment}[1]{}
\usepackage{xcolor}
\definecolor{ginger}{rgb}{0.69, 0.4, 0.0}

\newcommand{\lr}[1]{\left( #1\right)}
\newcommand{\mlr}[1]{\left[ #1\right]}

\newcommand{\alr}[1]{\left\langle #1\right\rangle}
\newcommand{\norm}[1]{\left\lVert#1\right\rVert}
\newcommand{\abs}[1]{\left\lvert#1\right\rvert}

\newcommand{\ii}{\mathrm{i}}
\newcommand{\ee}{\mathrm{e}}

\newcommand{\where}{\quad {\rm where}\quad}
\newcommand{\order}{\mathrm{O}}
\newcommand{\torder}{\tilde{\order}}

\newcommand{\OO}{\mathcal{O}}

\newcommand{\xalpha}{a}
\newcommand{\xbeta}{b}
\newcommand{\kzero}{\ket{\bm{0}}}
\newcommand{\cl}{_{\rm cl}}
\newcommand{\powerlaw}{\gamma}

\renewcommand{\thesection}{\arabic{section}}

\makeatletter
\renewcommand{\p@subsection}{}
\renewcommand{\p@subsubsection}{}
\makeatother

\usepackage{amsthm}
\newtheorem{thm}{Theorem}

\begin{document} 

\title{Fast and Accurate Greenberger–Horne–Zeilinger Encoding Using All-to-all Interactions}
\author{Chao Yin}\email{chao.yin@colorado.edu}
\affiliation{Department of Physics and Center for Theory of Quantum Matter, University of Colorado, Boulder CO 80309, USA}

\date{\today}

\begin{abstract}
    The $N$-qubit Greenberger–Horne–Zeilinger (GHZ) state is an important resource for quantum technologies. We consider the task of GHZ encoding using all-to-all interactions, which prepares the GHZ state in a special case, and is furthermore useful for quantum error correction, \revise{interaction-rate enhancement, and transmitting information using power-law interactions}. The naive protocol based on parallelizing CNOT gates takes $\order(1)$-time of Hamiltonian evolution.
    In this work, we propose a fast protocol that achieves GHZ encoding with high accuracy. The evolution time $\order(\log^2N/N)$ almost saturates the theoretical limit $\Omega(\log N/N)$. Moreover, the final state is close to the ideal encoded one with high fidelity $> 1-10^{-3}$, up to large system sizes $N\lesssim 2000$. The protocol only requires a few stages of time-independent Hamiltonian evolution; the key idea is to use the data qubit as control, and to use fast spin-squeezing dynamics generated by e.g. two-axis-twisting.
\end{abstract}

\maketitle

\emph{Introduction.}---
Quantum technologies offer great possibilities to perform information processing tasks far more efficiently than classical machines. 
For example, quantum computers are potentially able to factorize large numbers \cite{Shor99} or solve linear systems of equations \cite{HHL09} exponentially fast. As of estimating unknown parameters, quantum metrology can achieve a Heisenberg-limited precision that surpasses classical resources \cite{metro_science04}, using entangled states like the Greenberger–Horne–Zeilinger (GHZ) state \cite{GHZ89} and spin-squeezed states \cite{squeeze93,squeeze_rev11,squeeze_power_Ryd23,OAT_Ryd_Adam23,OAT_ion1d_23}. 
By far, quantum computation and squeezing physics are rather disconnected: Although squeezing implies entanglement \cite{squeeze_entan_prl05,squeeze_entan_pra06}, one usually does not care about the precise form of the squeezed state in the many-body Hilbert space: one merely uses a single parameter to describe the squeezing strength, which is sufficient to infer the ultimate precision in metrology \cite{squeeze_rev11}. In contrast, quantum computation aims to produce precise states from which one can deduce e.g. the precise factors of a large number. As a result, quantum computation is usually modeled by digitized quantum circuits, instead of analog Hamiltonian evolution that is more suitable for squeezing.

However, quantum circuits may not exploit the full power of current quantum platforms, because many of them are
naturally equipped by Hamiltonians with \emph{high connectivity}: Rydberg atoms interact with strength $r^{-\powerlaw}$ that decays as a power law with distance $r$, where the exponenent $\powerlaw=3,6$ \cite{Rydberg_rev20}; 
moreover, all-to-all interactions \begin{equation}\label{eq:all2allH=}
    H(t)=\sum_{i<j} J_{ij}^{\xalpha\xbeta}(t) X^\xalpha_i X^\xbeta_j, \where \abs{J_{ij}^{\xalpha\xbeta}(t)}\le 1,
\end{equation} 
\footnote{Here $H$ is normalized, and $X_i^\xalpha$ ($\xalpha=1,2,3$) are Pauli matrices on qubit $i$; we will also use notations $X_i,Y_i,Z_i$.} arise in optical cavities \cite{cavity_metro_scram23,cavity_graphstate24}, almost in trapped ions with $\powerlaw\in [0, 3]$ \cite{ion_longrange12,ion_rev21}, and are proposed for superconducting circuits \cite{all2all_supercond}. With such Hamiltonian, a single qubit interacts with many others simultaneously, which potentially enhances the information processing speed. This is particularly demanding in the current
noisy intermediate-scale quantum (NISQ) era, where one wants to perform computation faster than the decoherence timescale, when useful information would be lost into the environment. Indeed, for any power-law exponent $\powerlaw<2d+1$, protocols have been constructed that transmit quantum information through space in a way that is asymptotically faster than local quantum circuits/ Hamiltonians in $d$ spatial dimensions \cite{power_hierarchy,power_GHZ21,power_yifan21}. For $\powerlaw>d$, there are also lower bounds on the protocol time that matches the fast protocols \cite{power_chen19,power_KSLRB20,power_hierarchy,power_LRB21}; see the recent review \cite{ourreview}. 
However, such speed limits are much less understood for $\powerlaw<d$ \cite{Wprotocol_gorshkov20}, where the system is closer to the all-to-all limit. Here, the notion of spatial locality is challenged by a diverging local energy density; although progress have been made regarding Frobenius operator growth \cite{all2all_SYK20,all2all_yin20,power_KSOTOC21}, this is not directly related to tasks of e.g. preparing a certain entangled state.

In this work, we focus on the \revise{task to perform GHZ encoding in $N+1$ qubits $i\in\Lambda:=\{0,1,\cdots,N\}$}: For any quantum state $\alpha\ket{0}+\beta\ket{1}$ with $\abs{\alpha}^2+\abs{\beta}^2=1$ originally contained in the data qubit $i=0$, \revise{a GHZ encoding} unitary $U$ encodes the quantum data into the GHZ subspace of all qubits: \begin{equation}\label{eq:GHZencode}
U \ket{\alpha,\beta;0} \approx \ket{\alpha,\beta},\quad \forall \alpha,\beta
\end{equation}
where $\approx$ allows some error quantified shortly, and $\ket{\alpha,\beta;0}:= \lr{\alpha\ket{0}+\beta\ket{1}}_0\otimes \kzero,\ket{\alpha,\beta}:=\alpha\kzero_\Lambda+ \beta\ket{\bm{1}}_\Lambda$.
Here $\kzero$ ($\kzero_\Lambda$) for example is the all-zero state on qubits $\{1,\cdots,N\}$ (all qubits $\Lambda$). \revise{We ask the following question:} 

\vspace{5pt}

\revise{{\itshape What is the shortest evolution time $T$ for an all-to-all Hamiltonian evolution \eqref{eq:all2allH=} to achieve GHZ encoding $U$?}}

\vspace{5pt}

To the best of our knowledge, the previous fastest protocol 
is to apply Hamiltonian $H=Z_0(X_1+\cdots+X_N)$ for time $T=\pi/4=\Theta(1)$ \footnote{Throughout, we use big-O notations on the scaling at $N\rightarrow\infty$: $f=\order(g)$ ($f=\Omega(g)$) means $f\le cg$ ($f\ge cg$) for some constant $c$ independent of $N$, and $\Theta(\cdot)$ means both $\order(\cdot)$ and $\Omega(\cdot)$. Tildes in e.g. $\tilde{\mathrm{O}}$ means hiding polylogarithmic factors.}, before locally rotating all $i>0$ by $\ee^{-\ii\pi X_i/4}$. Here we assume local \revise{on-site} rotations are arbitrarily fast
, since they do not change entanglement, and are usually much faster than interactions in reality \cite{Rydberg_rev20,ion_rev21,cavity_Zeyang22,cavity_BCS24}; \red{see Supplemental Material (SM) \footnote{SM also contains 
the proof of \eqref{eq:lowerbound} adapted from \cite{ourreview}, and further details of our protocol.
} for further justification \revise{of the setup, where our protocol below is reformulated to contain no local rotations}}. Alternatively, this exact protocol can be viewed as applying CNOT gates in parallel with the data qubit as control. However, there is a large gap between this $\Theta(1)$ runtime and a lower bound \begin{equation}\label{eq:lowerbound}
    T=\Omega(\log N/N),
\end{equation} 
for generating $U$ by a possibly time-dependent Hamiltonian \eqref{eq:all2allH=},
where the latter vanishes quickly with $N$
. \red{Although this is because the above protocol only uses $\order(N)$ out of $\Theta(N^2)$ pairs of interactions
, it was not clear how to \revise{utilize the additional} couplings to speed up the evolution while maintaining the final state to be digital as in \eqref{eq:GHZencode}.

Here, we resolve this gap by providing a protocol (\eqref{eq:whole} below) with} runtime \begin{equation}\label{eq:T=log2N}
    T=\order(\log^2N/N),
\end{equation} 
that almost saturates the bound \eqref{eq:lowerbound}.
This is the first protocol with such a small runtime $T=\torder(1/N)$ for any digital quantum information processing task \footnote{To the best of our knowledge, the only previously-known protocol with a vanishing runtime at large $N$, is a $W$-state \cite{W_state} generation protocol \cite{Wprotocol_gorshkov20} with $T=\Theta(1/\sqrt{N})$.}. Our main idea is to generate many-body entanglement using fast spin-squeezing protocols like two-axis-twisting (TAT) \cite{squeeze93}, which generates extreme squeezing in short time $T=\torder(1/N)$. 
Although the squeezing subroutines make our protocol inexact, we carefully design ``unsqueezing'' stages that cancel the unwanted squeezing effects with high precision, bridging the gap between analog and digital quantum evolution. Remarkably, the error is very small for all system sizes $N\lesssim 2000$ studied numerically, quantified by the worst-case infidelity 
with respect to the target state
:\begin{equation}\label{eq:eps=}
    \epsilon := 1-\min_{\alpha,\beta} \abs{\bra{\alpha,\beta}U \ket{\alpha,\beta;0}}^2.
\end{equation}
We have $\epsilon< 10^{-3}$ when the numerical coefficient in \eqref{eq:T=log2N} is $\approx 1$, and can be further improved systematically by increasing $T$.

\emph{Implications.}---
In the special case $\alpha=\beta=1/\sqrt{2}$, our protocol prepares the GHZ state with high fidelity.
GHZ state generation have been studied theoretically \cite{GHZ_diffO1time19,GHZ_postselect20,GHZ_Jiazhonghu21,GHZ_OAT22,GHZ_cubic_intera23,GHZ_tear23,GHZlike24} and experimentally with $N\sim 20$ \cite{GHZ_ion11,GHZ_ion21,GHZ_photon18,GHZ_photon18_12,GHZ_SC19,GHZ_SC21,GHZ_Rydberg19,GHZ_Ryd24}. 
However, all existing protocols either take long evolution time $T=\Theta(1)$, or produce GHZ-like states, whose fidelity with the true GHZ state is not as high as ours 
at large $N\gtrsim 10^3$. For example, \cite{GHZlike24} proposes a $T=\order( \log N/N)$ protocol 
that prepares a GHZ-like state of the form $c_0\lr{\ket{D^N_0} + \ket{D^N_N}}/\sqrt{2} + c_1 \lr{\ket{D^N_1} + \ket{D^N_{N-1}}}/\sqrt{2} + \cdots$, where $\ket{D^N_k}$ is the Dicke state that is the equal superposition of all computational basis states with $k$ ones and $N-k$ zeros (so $(\ket{D^N_0} + \ket{D^N_N})/\sqrt{2}$ is the GHZ state). \cite{GHZlike24} achieves $\epsilon = 1-\abs{c_0}^2\approx 0.65$ for $N\sim 100$, 
so when used for metrology, this GHZ-like state may not be as clean as the exact one \cite{HL_sampling24}. Moreover, it only prepares one state and does not achieve GHZ encoding for any $\alpha,\beta$. 
Nevertheless, the spin-squeezing idea in \cite{GHZlike24} inspired our work that overcomes the above challenges. 

Beyond preparing GHZ states, our protocol can be used to enhance interaction rate: 
Given two qubits, one can first encode each of them into the GHZ subspace of a qubit ensemble of size $\sim N$, and then the two ensembles can interact with a strength $\sum_{ij}Z_i Z_j\sim N^2\gg 1$ due to their large polarizations. By reversing the GHZ encoding procedure (i.e. uncomputing)  and adding fast single-qubit rotations, this implements an arbitrary two-qubit gate in time $2T+\order(1/N^2)=\torder(1/N)$. This is much faster than $\Omega(1)$ time for directly turning on the interaction between the two. The gate error $\sim \epsilon/N$ is also \red{extremely} small, as we will see. 

The GHZ subspace can be viewed as a quantum error correction (QEC) code that corrects bit-flip errors with maximal distance $N+1$
. 
Concatenating such spin-cat codes \cite{spincat_code24} as in the Shor code \cite{shor_code95} also corrects phase errors, leading to a high threshold for quantum fault-tolerance \cite{spincat_code24}. Our protocol can be used to encode/decode such QEC codes in a fast way, where decoherence errors have little time to accumulate during the process \cite{GHZlike24}; \revise{see SM for a quantitative analysis}. 
As a remark, there may also exist ``gate'' errors due to imprecise control in the protocol, and a small change \red{$\delta T=\tilde{\Theta}(1/N)$} in evolution time \revise{that is comparable to the whole $T$} is already detrimental. As a result, there is a tradeoff between these different error types, and one may want to optimize over the parameters when implementing our protocol in reality.

\revise{Our result strongly suggests that $\tilde{\Theta}(N^{\frac{\gamma}{d}-1})$ time is sufficient to transmit information between any pair of qubits using interactions bounded by a power law $\abs{J_{ij}^{\xalpha\xbeta}(t)}\le r_{ij}^{-\gamma}$, because one can just apply our GHZ encoding protocol \eqref{eq:all2allH=} normalized by the weakest coupling strength 
at $r_{ij}=N^{1/d}$. Although our protocol is not exact, the transmitted signal approaches unity at larger system size in numerics (see SM). If this continues to hold in the asymptotic limit $N\rightarrow \infty$, this would saturate the signaling bound derived in \cite{Wprotocol_gorshkov20} for $\gamma<d$ and essentially close the problem of generalizing Lieb-Robinson bounds to power-law interacting systems \cite{power_LRB21,ourreview} (see Fig.~\ref{fig:cubic}(b)). }

\begin{figure}
    \centering
    \includegraphics[width=0.45\textwidth]{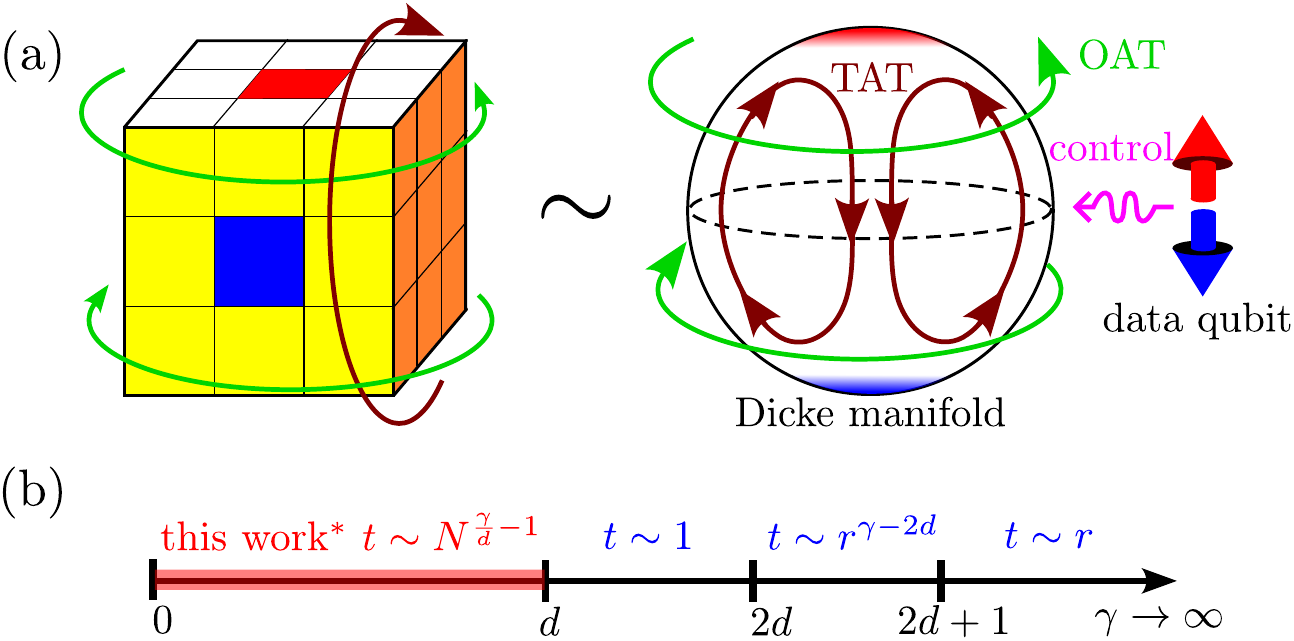}
    \caption{\revise{(a)} Our GHZ encoding protocol (right) shares the similar spirit of the Rubik's Cube (left), where 
    complicated tasks are decomposed to a set of fast basic operations \revise{including spin-squeezing dynamics}. 
    \revise{(b) Fastest possible signaling time $t$ across distance $r$ in power-law interacting systems, where polylogarithmic factors in $r,N$ are ignored \cite{power_hierarchy,power_GHZ21,power_yifan21,power_chen19,power_KSLRB20,power_LRB21,Wprotocol_gorshkov20,ourreview}. Our protocol strongly suggests (but does not rigorously prove) the speed limit shown in red for power-law exponent $\gamma<d$.}
    }
    \label{fig:cubic}
\end{figure}

\emph{The protocol and building blocks.}---
Our GHZ encoding unitary $U$ 
consists of a few simple subroutines:
\begin{equation}\label{eq:whole}
    U = S_{\tau_3} \lr{R^Z_{\pi/4} O_{\pi/4} R^X_{-\pi/2}} S_{-\tau_2}\, \lr{C_\phi\, S_{\tau_1}},
\end{equation}
where $\tau_p\,(p=1,2,3)$ and $\phi$ are 
parameters that we tune and optimize. The high-level idea of \eqref{eq:whole} is analogous to Rubik's Cube shown in Fig.~\ref{fig:cubic}(a): we first identify the basic operations, and then decompose complicated tasks into them. In our problem, the GHZ encoding task is decomposed to \emph{fast} basic operations defined by
\begin{subequations}
    \begin{align}
    &\red{R^{X}_{\phi} :=\, \ee^{\ii \frac{\phi}{2} X}, \quad R^{Z}_{\phi} :=\, \ee^{\ii \frac{\phi}{2} Z},} & C_\phi :=&\, \ee^{\ii \frac{\phi}{2} Z_0\otimes X}, \label{eq:controlrotation} \\
    &S_\tau:=\, \ee^{-\ii \tau \frac{\log N}{N} H_{\rm TAT}}, & O_\phi :=&\, \ee^{-\ii \frac{\phi}{4N} Z^2}, \label{eq:OAT}
\end{align}
\end{subequations}
where $H_{\rm TAT}=XY+YX$, and \red{$X:=\sum_{i=1}^N X_i$ (excluding site $0$; similar for $Y,Z$)}
.
It is useful to consider them as acting on the Dicke manifold (DM) $\mathrm{Span}(\ket{D^N_k}:k=0,1,\cdots,N)$ of the $N$ qubits excluding $0$, which at $N\rightarrow\infty$ becomes a semiclassical phase space: a sphere. For example, $R^Z_{\phi}$ is just angle-$\phi$ rotation from east to west on the sphere, where we set $Z=N$ as the north pole.
$C_\phi$ is controlled $X$-rotation, where the sphere rotates in opposite directions depending on whether the control qubit $i=0$ is in $\ket{0}$ or $\ket{1}$
. The TAT unitary $S_\tau$ generates spin squeezing 
for initial product state $\kzero$. In particular,  extreme squeezing occurs at time $\tau=\tau_{\rm min}\approx1/8$ \cite{OAT2TAT}, which can be understood by the semiclassical trajectories as reviewed in SM. $O_\phi$ is the one-axis twisting (OAT) unitary \cite{squeeze93} 
that can also generate squeezing
. However, here it is more useful to interpret it as a relative angle-$\phi$ rotation that rotates the north and south semispheres in opposite directions. 

\emph{Combining building blocks.}---
We use a normalized parameter $\theta$ to express the controlled-rotation angle $\phi$: \begin{equation}\label{eq:thetalogN}
    \phi = \theta \lr{\log^2N}/N.
\end{equation}
As we will see, the parameters $\tau_p,\theta$ do not scale with $N$, so the total time $T$ of our protocol \eqref{eq:whole} is dominated by $\theta$: \begin{equation}\label{eq:T=theta}
    T=\theta \frac{\log^2N}{2N}+ \frac{\pi}{8N} + \frac{\log N}{N}\sum_{p=1,2,3} \tau_p \approx \theta \frac{\log^2N}{2N},
\end{equation}
at sufficiently large $N$. 
We have combined the subroutines into four stages in \eqref{eq:whole}:

\emph{1. Squeeze-to-separate $C_\phi\, S_{\tau_1}$:} In the ultimate encoded state $\ket{\alpha,\beta}=\alpha \ket{0}_0\otimes \kzero + \beta \ket{1}_0\otimes \ket{\bm{1}}$, the polarization of the $N$ qubits depend drastically on the state of the control qubit $0$
. In other words, the two parts ($\alpha$-part and $\beta$-part) of the state are supported in faraway regions on the DM for the $N$ qubits. Since the two parts have the same initial support $\kzero$, the protocol needs to first \emph{separate} them to \emph{disjoint} regions, before pulling the two supports faraway from each other. The naive way to separate is just a controlled rotation: $C_\phi \lr{\alpha\ket{0}+\beta\ket{1}}_0\otimes \kzero = \alpha\ket{0}_0\otimes \ket{\bm{\phi}}+\beta\ket{1}\otimes \ket{\bm{-\phi}}$; see \cite{similar_separate_idea16} for similar ideas. However, the spin-coherent states $\ket{\bm{\pm \phi}}$ have quantum fluctuations $\Delta Y\sim\sqrt{N}$, so they have disjoint support on the size-$N$ DM only after $\phi\gtrsim 1/\sqrt{N}$, which would lead to a final $T\gtrsim 1/\sqrt{N}$ \footnote{We expect that this naive separation strategy (with the latter stages adjusted accordingly) also yields an approximate GHZ encoding protocol with high accuracy, similar to what we will show for \eqref{eq:whole}. The $T=\torder( 1/\sqrt{N})$ runtime still beats the known protocols with $T=\Theta(1)$. }. 

Here $C_\phi\, S_{\tau_1}$ first squeezes the state using TAT to reduce the quantum fluctuation, and then controlled-rotates, as shown in Fig~\ref{fig:idea}(a)
. Naively, at $\tau_1=\tau_{\rm min}\approx 1/8$ the squeezing is extreme $\Delta Y\sim 1$, so rotation angle $\phi\sim (\log N)/N$
suffices. Here the $\log N$ factor is such that the distance between the two supports is $\log N$-times larger than their fluctuation width, so that their overlap is inverse-polynomially small in $N$ (assuming the decay over distance is Gaussian for example). However, we add one more factor of $\log N$ in \eqref{eq:thetalogN}. The reason is that at extreme squeezing, the wave packet is also extremely stretched along the perpendicular direction: $\Delta X \sim N$, and it is hard to refocus such an expanded wave packet back to spin-coherent states
. Therefore, we desire $\tau_1<\tau_{\rm min}$ so that the state has not evolved to extreme squeezing. In this case, numerics in SM shows that $\Delta Y\sim \log N$ for nonvanishing but small $\abs{\tau-\tau_{\rm min}}$
, leading to the extra $\log N$.

\emph{2. Pulling-away $S_{-\tau_2}$:} 
We then pull the two parts far away from each other until they become antipodal on the sphere. This can be done in a fast way again using 
TAT. Focusing on the centers of the two green regions in Fig.~\ref{fig:idea}(a), they are initially separated in the $y$ direction with distance $\approx 2N \phi$, so we want to reverse the direction of the TAT dynamics 
to stretch (instead of squeeze) the $y$ direction. After time $\tau_2>0$, $S_{-\tau_2}$ maps the two wave-packet centers to the $Y=\pm N$ antipodal points, while the two states become squeezed in the $x$ direction, as shown by blue in Fig.~\ref{fig:idea}(a). Intriguingly, the first two stages above are analogous to the method of signal amplification using a time-reversed interaction (SATIN) \cite{SATIN16,SATIN16_Dur,SATIN_robust17,SATIN_exp22}, where the ``signal'' that $C_\phi$ imprints on the $N$ qubits is amplified by the squeezing unitaries that sandwich it, \revise{thereby requiring a
shorter evolution during $C_\phi$}.

\begin{figure}
    \centering
    \includegraphics[width=0.45\textwidth]{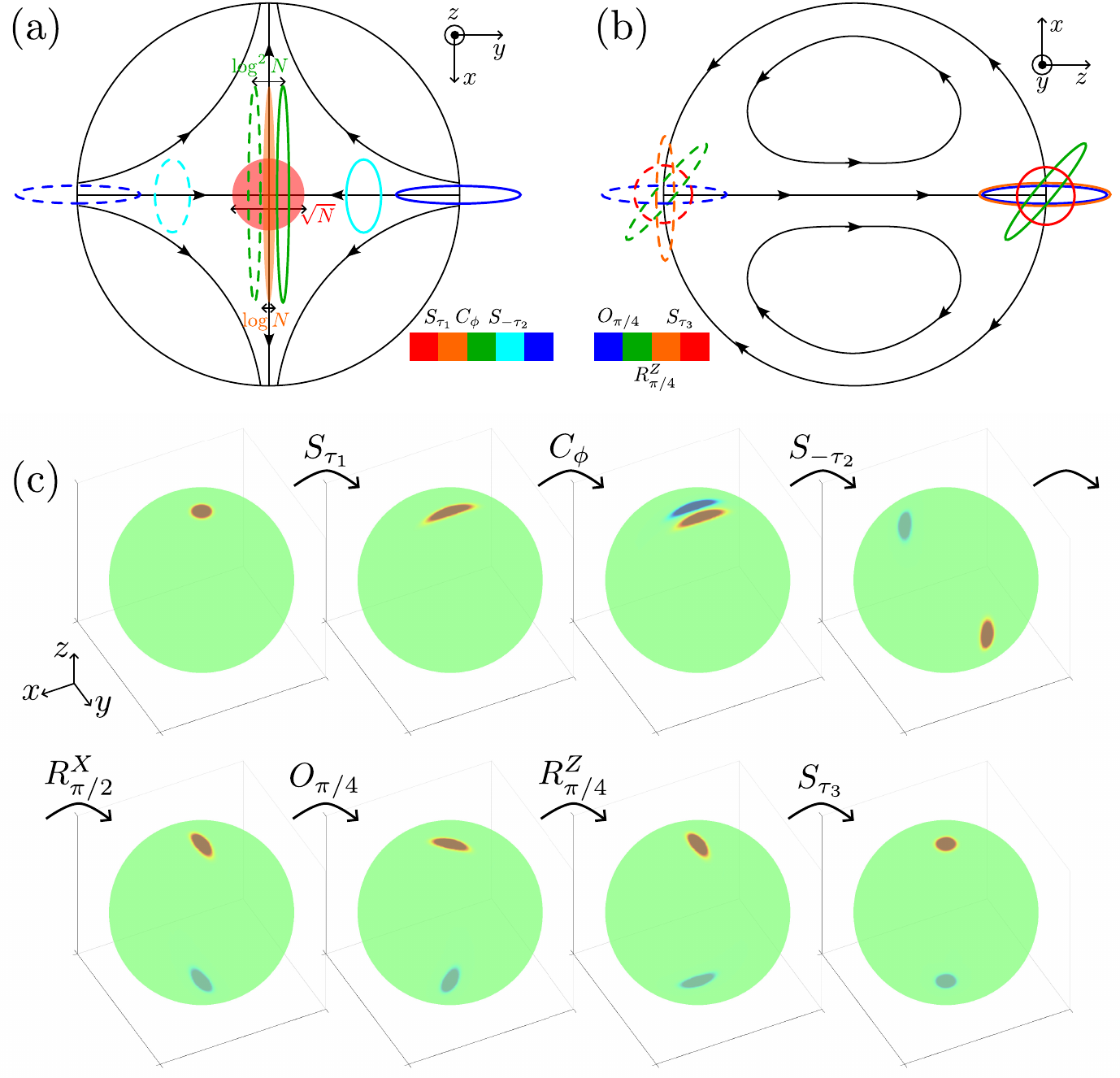}
    \caption{(a,b) Sketch of our protocol on the $N$-qubit DM, with the first stages in (a) and latter stages in (b). The colors represent different stages. For example, the red shaded circle represents the initial state $\kzero$, which is evolved by $S_{\tau_1}$ to the orange shaded oval representing a squeezed state. After a controlled rotation the wave function separates to two parts shown by solid ($\alpha$-part) or dashed ($\beta$-part) borderlines
    . The black lines show classical trajectories of $H_{\rm TAT}$ with positive $\tau$. States at the boundary of the shown semisphere 
    have parts outside the boundary, which should be thought of as folding onto the other semisphere. 
    The final (blue) stage of (a) is rotated by $R^X_{-\pi/2}$ to get the first (blue) stage of (b). 
    (c) Husimi distribution of the state on the DM at different stages in our protocol, from numerics with parameters $(N,\theta,\tau_1,\tau_2,\tau_3)=(1024,2,0.0505, 0.111, 0.0357)$. Starting from the third plot, the distribution is shown separately for the two parts using red($\alpha$) and blue($\beta$). More precisely, the color (red$\rightarrow$green$\rightarrow$blue) quantifies $Q_\alpha-Q_\beta$ (positive$\rightarrow0\rightarrow$negative) with $Q_\alpha$ being the Husimi Q function for the $\alpha$-part $\ket{\psi}$. 
    }
    \label{fig:idea}
\end{figure}

\emph{3. Rotations $R^Z_{\pi/4} O_{\pi/4} R^X_{-\pi/2}$:} We then rotate the two antipodal points to north and south poles by $R^X_{-\pi/2}$, so that we have already obtained a GHZ-like encoded state where the two parts of the state are squeezed states instead of spin-coherent ones like $\kzero$. It remains to ``unsqueeze'' them. However, we cannot use TAT directly, because if it unsqueezes the state at the north pole for example, it will further squeeze the south-pole state. This is because the two states are both squeezed in the $x$ direction, while TAT has different squeezing directions at the two poles, as one can see from the blue states and black trajectories in Fig.~\ref{fig:idea}(b) \footnote{The $3$-local Hamiltonian in \cite{GHZlike24} has the same squeezing directions at the two poles, so in principle can be used here to unsqueeze. However, to engineer the effective $3$-local interaction, the $2$-local Hamiltonian \eqref{eq:all2allH=} needs to change very rapidly. }. Therefore, before unsqueezing, we first relatively-rotate the two states by $\pi/2$ using $O_{\pi/4}$ (blue $\rightarrow$ green), and then rotate by $R^Z_{\pi/4}$ to align them back to the $x,y$ directions (green $\rightarrow$ orange).

\emph{4. Unsqueeze $S_{\tau_3}$:} Finally, after aligning the two states with the TAT trajectories, we unsqueeze them by $S_{\tau_3}$ with an optimal $\tau_3$. Although the final states are not perfect $\kzero,\ket{\bm{1}}$, the error is small because their support can be made very close to a circle region of minimal size, as shown by red in Fig.~\ref{fig:idea}(b).

\emph{Performance of the protocol.}---
To demonstrate the above ideas and quantify the performance, we simulate the system numerically up to large $N= 2048$
. 
Denoting the final state as $\alpha\ket{0}_0\otimes \ket{\psi}+\beta\ket{1}_0\otimes \cdots $, the $\alpha$-part $\ket{\psi}$ alone determines the infidelity $\epsilon = 1-\abs{\alr{\bm{0}|\psi}}^2$
due to a symmetry of our protocol (see SM).
For a given $N$, we sweep the parameter regime $\tau_1,\tau_2\in[0,0.15], \theta\in[0,2]$, and for each set of these parameters, we numerically optimize $\tau_3\in [0,0.15]$ in the last unsqueezing stage, to maximize the overlap $\abs{\alr{\bm{0}|\psi}}$. 
For a typical set of parameters with $N=1024,\theta=2$ and $\tau_p$s optimized accordingly, we show in Fig.~\ref{fig:idea}(c) that the evolution indeed follows the previous intuitions
. Moreover, the infidelity is tiny $\epsilon\approx 6.7\times 10^{-4}$! Putting numbers in \eqref{eq:T=theta} yields $T=0.048$ for this parameter set, which is only $\sim 6\%$ of the parallelizing-CNOTs protocol. 

\begin{figure}
    \centering
    \includegraphics[width=0.45\textwidth]{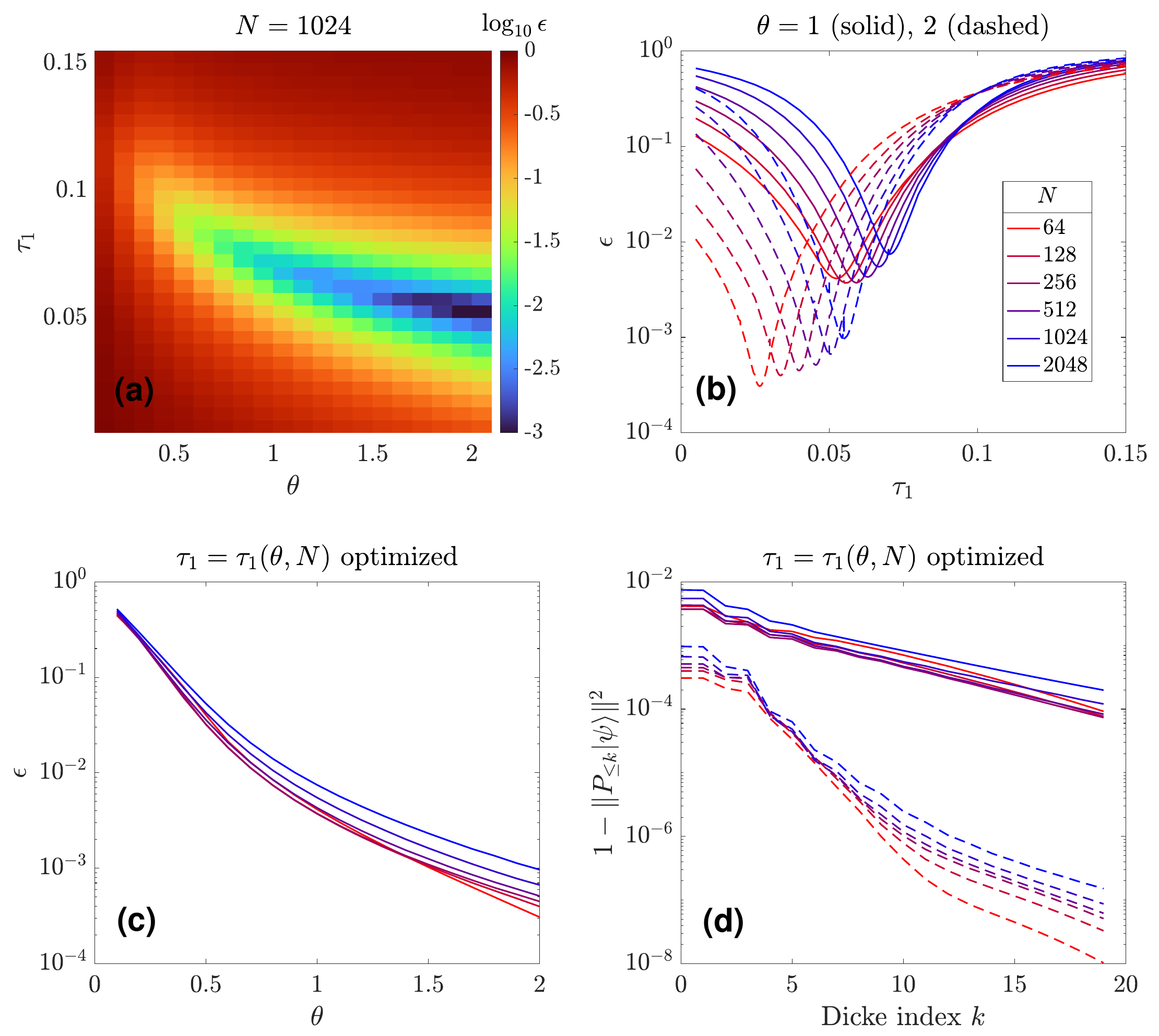}
    \caption{Performance of our protocol where $\tau_2,\tau_3$ have been optimized. (a) $\log_{10}\epsilon$ as a function of $\theta$ and $\tau_1$ for $N=1024$. (b) Solid (dashed) lines: Fixing $\theta=1$ ($\theta=2$), infidelity $\epsilon$ as a function of $\tau_1$, where different color represents different $N$ (same below)
    . The function is minimum at $\tau_1=\tau_1(\theta,N)$. (c) With $\tau_1$ fixed to the optimal $\tau_1(\theta,N)$, $\epsilon$ decays exponentially with $\theta$. (d) The support of the $\alpha$-part final state $\ket{\psi}$ over Dicke states decays with $k$ exponentially, where solid (dashed) lines are still for $\theta=1$ ($\theta=2$), and $\tau_1=\tau_1(\theta,N)$ is optimized. }
    \label{fig:sweep}
\end{figure}

We then study how the performance of our protocol depends on the parameters.
The optimized $\tau_3$s are reported in SM, where we also show that the optimal $\tau_2$ turns out to be independent of $\tau_1$, which can be understood from semiclassical trajectories. 
Fig.~\ref{fig:sweep} then shows the dependence on the remaining parameters $N,\theta,\tau_1$, where $\tau_2,\tau_3$ have been optimized accordingly. 
For a given $N$ and $\theta$, there exists an optimal squeezing time $\tau_1=\tau_1(\theta, N)$. 
With $N$ fixed, Fig.~\ref{fig:sweep}(a) shows that 
$\tau_1(\theta, N)$ decreases with increasing $\theta$ with a slower and slower slope. The intuition is that 
with less squeezing in the first stage, the controlled-rotation angle needs to be larger to fully separate the two wave packets; if $\tau_1$ is too small $\ll 1$, the $N$-scaling of $\phi$ in \eqref{eq:thetalogN} would be insufficient, and in the extreme case $\phi=\pi/2$ the protocol reduces to the parallelizing-CNOT protocol without squeezing. 

At fixed $\theta$, Fig.~\ref{fig:sweep}(b) shows that $\tau_1(\theta, N)$ drifts to larger values when $N$ increases. We expect this to be a finite-size effect, because $\tau_1$ is upper bounded by $\tau_{\rm min}$ anyway where squeezing is extreme. The infidelity at $\tau_1(\theta, N)$ also slowly increases with $N$, and it is unclear from the finite-size numerics whether it would saturate at some $<1$ value at $N\rightarrow\infty$. Therefore, our protocol is not asymptotically exact, i.e. $\epsilon$ as a function of $N$ does not tends to zero for a fixed $\theta$. Nevertheless, Fig.~\ref{fig:sweep}(c) shows that the infidelity decays \emph{exponentially} when increasing $\theta$, so one can adjust $\theta$ a little bit to compensate the infidelity increase with $N$. To understand this exponential decay, observe that 
the unsqueezing stage cannot cancel the previous squeezing effects perfectly so that error occurs; at larger $\theta$, a smaller squeezing $\tau_1$ is sufficient and is easier to cancel.
In Fig.~\ref{fig:sweep}(d), we also plot the distribution of the $\alpha$-part $\ket{\psi}$ over Dicke states quantified by $\norm{P_{\le k}\ket{\psi}}^2 := \sum_{\ell=0}^k \abs{\alr{D^N_\ell|\psi}}^2$,
showing an exponential decay of support on $\ket{D^N_k}$ with large $k$. As a result, when using our protocol to enhance qubit interactions, the logical error $\sim \epsilon/N$ \red{is further suppressed by $N$} as advertised, because $\norm{P_{\le 0}\ket{\psi}}^2=1-\epsilon$, so that $\ket{\psi}$ has polarization $\alr{Z}_\psi=\sum_k (N-k)\lr{\norm{P_{\le k}\ket{\psi}}^2-\norm{P_{\le k-1}\ket{\psi}}^2}=N-\order(\epsilon)$ with relative error $\order(\epsilon/N)$; \revise{see SM for a plot of $\alr{Z}_\psi$ whose relative error \emph{decreases} with $N$}.

\emph{Outlook.}---
In this work, we propose a fast and accurate protocol \eqref{eq:whole} for GHZ encoding using all-to-all interactions, opening up the potential for QEC and interaction enhancement in (NISQ) quantum devices. The protocol uses subroutines feasible for experiments: OAT has been realized in many quantum platforms \cite{OAT_ion00,OAT_ion08,OAT_BEC10,OAT_BEC10_1,OAT_cavity10,OAT_cavity18,GHZ_SC19,OAT_optical23,OAT_Ryd_Adam23}, while TAT has been recently demonstrated \cite{TAT_haoqing24,TAT_JunYe24}. The controlled rotation (in a different local basis) can be decomposed as two OAT stages where \revise{the first OAT acts on $N+1$ qubits, and the second OAT (with an overall
negative sign) acts on $N$ qubits to cancel the dynamics within them; this} may require moving the control qubit faraway during the evolution. On the other hand, it may be more natural to let a global bosonic mode like photons to control an ensemble of qubits; in this way our protocol encodes a bosonic quantum state.  

For future directions, it is interesting to see whether fine-tuning the time-dependence of $H(t)$ can further improve the fidelity: Given that the error of our protocol is already small, we conjecture the existence of a protocol with similar time \eqref{eq:T=log2N}, and an infidelity $\epsilon$ that \emph{vanishes} at large $N$.
This \revise{would solve the power-law speed limit problem if true, and} may require a more systematic semiclassical analysis \cite{cohe_mathphy21}. 
It is also worth generalizing our protocol \revise{to more realistic settings} beyond the uniform all-to-all case. In SM, we take a first step and show that the protocol \revise{remains nearly perfect for moderate system size $N\lesssim 20$ even when the coupling coefficients are inhomogeneous with $\sim 20\%$ fluctuation, a scenario}
relevant in experiments 
\cite{inhomo_Ryd14,inhomo_spinEcho16,inhomo_Ryd23}. Furthermore, 
\revise{to generalize to \emph{truly decaying} power-law interactions,} recent developments on spin squeezing in such systems \cite{squeeze_power_Rey20,squeeze_power_PRA22,squeeze_yao23,rotor_spinwave23} would be useful. \revise{Finally, our work opens up possibilities to leverage spin squeezing to accelerate other quantum tasks, e.g. preparing the $W$-state.}

\emph{Acknowledgements.}---
We thank Andrew Lucas, Ana Maria Rey, David T. Stephen and Haoqing Zhang for valuable discussion. This work was supported by the Department of Energy
under Quantum Pathfinder Grant DE-SC0024324.

\bibliography{thebib}

\onecolumngrid

\newpage

\setcounter{equation}{0}
\setcounter{figure}{0}
\setcounter{page}{1}
\setcounter{section}{0}
\renewcommand{\theequation}{S\arabic{equation}}
\renewcommand{\thefigure}{S\arabic{figure}}
\renewcommand{\thesection}{S\arabic{section}}
\renewcommand{\thepage}{S\arabic{page}}

\begin{center}
    {\large \textbf{Supplementary Material: Fast and Accurate Greenberger–Horne–Zeilinger Encoding Using All-to-all Interactions}}
\end{center}

\section{Further justification of the setup}

\revise{Our $\tilde{\Theta}(1/N)$-time protocol may cause confusion at first glance: If one can already achieve GHZ encoding using standard protocols with $\Theta(1)$ time independent of system size $N$, why should one try a protocol with time that decays with $N$? We argue that there are indeed many good reasons to have a such faster protocol: 

First of all, the time scalings here are all based on the specific normalization of the Hamiltonian we choose, namely each pair of qubits interacts with an $\order(1)$ interaction strength $J_{ij}$. As will be discussed shortly, all-to-all interactions in e.g. trapped ions naturally have a different scaling $J_{ij}\sim 1/N$. In this case, the standard protocol using CNOT gates takes $t=\Theta(N)$ actual time that drastically increases with $N$, so it is a big deal to improve to $t=\tilde{\Theta}(1)$ time, which is achieved by our protocol. We will expand on this normalization issue in Section \ref{sec:normalization}.

Second, even if $J_{ij}$ does not scale with $N$ as in some cavity systems so that our protocol takes time $t\sim 1/N$, its much shorter evolution time leads to a much smaller decoherence error accumulated during the process; see Section \ref{sec:decoh} for a detailed discussion. In particular, the decoherence rate could scale with $N$ because we are preparing a globally entangled GHZ state, so it is particularly vital to shorten the protocol time. As an example mentioned in the main text, for $N\approx 1000$ qubits our protocol takes only $\sim 6\%$ evolution time of the CNOT protocol, making the final decoherence error potentially suppressed beyond one order of magnitude.

Third, as GHZ encoding maps a bit of local information to the global system, it can serve as simple and useful subroutines in a more complicated quantum information processing task: As discussed in the main text, our protocol can be used to enhance interaction rates and encode/decode quantum error correction codes. If the task calls the GHZ encoding subroutine many times, it is demanding to make each GHZ encoding subroutine as fast as possible even if the GHZ encoding time decreases with $N$, so that the total time of the whole task remains under control.

Fourth, as long as one goes beyond the all-to-all case and consider power-law interactions with exponent $\gamma>0$, all previous protocols have a time that grows with system size (when the nearest-neighbor $J_{ij}\sim 1$ is normalized). For example, the CNOT protocol will be constrained by the weakest coupling strength $J_{ij}\sim N^{-\gamma/d}$ at $r_{ij}\sim N^{1/d}$ in $d$ spatial dimensions, making the protocol time $t\sim N^{\gamma/d}$ growing with $N$. In contrast, our protocol makes it possible to do GHZ encoding in again $\order(1)$ (actually $\tilde{\order}(N^{\frac{\gamma}{d}-1})$) time for all $\gamma<d$.

Due to the above reasons, it is desirable to have an all-to-all GHZ encoding protocol that is as fast as possible.
}

\subsection{Arbitrarily fast local rotations}

We have assumed local rotations like $R^Z_\phi$ are instantaneous throughout, which we further justify here.

Theoretically, this is a natural assumption because local rotations do not change entanglement: One can always go to an interaction picture that eliminates all local rotations during evolution, by changing the time-dependent coefficient $J_{ij}^{\xalpha\xbeta}(t)$ of the all-to-all interaction. The final state will be equivalent to the target entangled state up to a local-basis change. \revise{For example, our protocol can be rewritten as \begin{align}\label{eq:whole1}
    U &= \lr{R^Z_{\pi/4} R^X_{-\pi/2}} \lr{R^X_{\pi/2} R^Z_{-\pi/4} S_{\tau_3} R^Z_{\pi/4}R^X_{-\pi/2} } \lr{R^X_{\pi/2} O_{\pi/4} R^X_{-\pi/2}} S_{-\tau_2}\, C_\phi\, S_{\tau_1} \nonumber\\
    &=  \lr{R^Z_{\pi/4} R^X_{-\pi/2}} \ee^{-\ii \tau_3 \frac{\log N}{N} \lr{AB+BA} } \ee^{-\ii \frac{\pi}{16N} Y^2} S_{-\tau_2}\, C_\phi\, S_{\tau_1},
\end{align}
where we have defined $A=(X+Z)/\sqrt{2},B=(X-Z)/\sqrt{2}$ to be polarizations along tilted axes, and used e.g. \begin{equation}
    R^X_{\pi/2} O_{\pi/4} R^X_{-\pi/2}= \exp\lr{-\ii \frac{\pi}{16N} R^X_{\pi/2}Z^2R^X_{-\pi/2}} = \exp\lr{-\ii \frac{\pi}{16N} R^X_{\pi/2}Z R^X_{-\pi/2}\cdot R^X_{\pi/2}ZR^X_{-\pi/2}}= \exp\lr{-\ii \frac{\pi}{16N} Y^2}.
\end{equation}
So without the final on-site rotations $R^Z_{\pi/4} R^X_{-\pi/2}$, our protocol $\lr{R^Z_{\pi/4} R^X_{-\pi/2}}^{-1}U$ achieves GHZ encoding (in a rotated local basis) using purely all-to-all interactions Eq.(1) in the main text. Furthermore, the Hamiltonian is still piecewise-constant in time.
} 

There is also a trick to implement fast rotations without possibly introducing complicated time-dependence in the coefficients: Suppose the system actually contains $2N$ qubits, then rotations on the first $N$ qubits, $R^Z_\phi$ for example, can be implemented by an interaction Hamiltonian $H = \sum_{i=1}^{N} Z_i \otimes \sum_{j=N+1}^{2N}Z_j$ for time $t=\phi/(2N)$, where the state on the last $N$ qubits is fixed in the all-zero state. In this way, by a constant space overhead $N\rightarrow 2N$, all local rotations are realizable in time $\order(1/N)$, which is subdominant than the time $\Omega(\log N/N)$ for GHZ encoding protocols.

In reality, the above trick may be subtle to implement, and one may want to use time-independent Hamiltonians with fixed axis (e.g. one may realize OAT only along the $Z$-axis), so it is probably most convenient to realize the rotations directly by single-qubit gates. Nevertheless, single-qubit gates are usually operated at much faster timescales than interactions. Sometimes this separation of timescale even scales with $N$: In trapped ions for example, all-to-all interactions are generated by coupling ions to the center-of-mass phonon mode \cite{ion_rev21}, so are necessarily suppressed by a factor of $1/N$ due to the $\propto 1/\sqrt{N}$ amplitude of the mode on a given site. In this setting, our protocol has actual evolution time $\torder(1)$, which holds even if the rotations are counted as constant time. In cavity systems, the interaction strength may also decrease with $1/N$ \cite{cavity_Zeyang22,cavity_metro_scram23}. Moreover, the on-site field can be made much larger than even $N\times $ interaction strength \cite{cavity_BCS24}. As a result, in these systems rotations can be effectively viewed as instantaneous.

\subsection{Normalization of the Hamiltonian and evolution time}\label{sec:normalization}

\revise{
At the beginning of this Supplemental Material, we have argued that it is important to search for a GHZ encoding protocol that is as fast as possible, in the space of all protocols generated by an all-to-all Hamiltonian Eq.~(1) in the main text. In this work we construct an explicit protocol that nearly saturates the lower bound (given by Theorem \ref{thm:bound} below), which is faster by a factor of $\sim N$ comparing to the standard CNOT protocol using Hamiltonian \begin{equation}\label{eq:HCNOT}
    H_{\rm{CNOT}}= Z_0(X_1+\cdots + X_N).
\end{equation}

Observe that our protocol Hamiltonian $H(t)$ has an operator norm $\norm{H(t)}\sim N^2$ that is much larger than $\norm{H_{\rm{CNOT}}}\sim N$; if one restricts to Hamiltonians with the same norm $\sim N$, then our protocol Hamiltonian should be normalized as $\frac{1}{N}H(t)$ and no longer provides an advantage over the CNOT protocol anymore. However, this comparison based on the same Hamiltonian norm is not fair: In $H_{\rm{CNOT}}$, the qubit $0$ couples to all the other qubits just as the all-to-all case, so if one can engineer $H_{\rm{CNOT}}$ in a physical system with normalization $J_{0i}\sim 1$, it would be very strange if one cannot obtain an all-to-all interaction among all pairs of qubits with the same coupling strength $J_{ij}\sim 1$. In other words, it is more natural to compare our protocol with \eqref{eq:HCNOT} using the same coupling strength \emph{per pair}.

More generally, we can compare our all-to-all protocol with other known GHZ encoding protocols based on any interaction connectivity (beyond just \eqref{eq:HCNOT}), where one indeed needs to take the Hamiltonian normalization factor into account. Our protocol can still be superior in such cases. To illustrate this point, we focus on the concrete setting of trapped ions. As discussed above, one can engineer an all-to-all interaction among the ions with strength \begin{equation}
    J_{ij}\sim 1/N,
\end{equation}
that decays with $N$, because the ions are only coupled to the single center-of-mass phonon mode. On the other hand, one can couple the ions to more modes to engineer a fairly local Hamiltonian $H_{\rm local}$ in $d$ spatial dimensions with nearest-neighbor coupling $J_{ij}\sim 1$ that does not decay with $N$ \cite{ion_rev21}. Here the two cases do have the same Hamiltonian norm $\sim N$ . Nevertheless, our all-to-all GHZ encoding protocol takes time $t\sim 1$, much faster than any local Hamiltonian evolution that would take time $t_{\rm local}=\Omega(N^{1/d})$. The reason is that information propagates with a bounded velocity under local Hamiltonian $H_{\rm local}$; this is in sharp contrast to \eqref{eq:HCNOT}, which has the same Hamiltonian norm but does not have spatial locality.

}

\section{Proof of the lower bound}
Proposition 9.3 in \cite{ourreview} derives the lower bound Eq.~(3) in the main text for exact protocols $\epsilon=0$; it is straightforward to generalize the proof to approximate cases:
\begin{thm}[adapted from \cite{ourreview}]\label{thm:bound}
For any constant $\delta>0$,
    GHZ encoding with infidelity \begin{equation}\label{eq:eps=1/2-delta}
        \epsilon= 1/2-\delta,
    \end{equation} 
    requires \begin{equation}\label{eq:T>logNN}
        T=\Omega\lr{\frac{\log(\delta N)}{N}}.
    \end{equation}
\end{thm}
\begin{proof}
\cite{ourreview} proves that \eqref{eq:T>logNN} is necessary for \begin{equation}\label{eq:UXUZ>1}
    \norm{[UX_0 U^\dagger, Z_1]}\ge 4\delta,
\end{equation}
so we only need to show that any GHZ encoding protocol with \eqref{eq:eps=1/2-delta} satisfies \eqref{eq:UXUZ>1}.

Define projector $P_{\alpha,\beta}:=\ket{\alpha,\beta;0}\bra{\alpha,\beta;0}$. By direct computation, we have \begin{align}\label{eq:compute_UXUZ}
    \bra{-\beta,\alpha}UX_0 U^\dagger Z_1 \ket{\alpha,\beta} &= \bra{-\beta,\alpha}U\mlr{P_{-\beta,\alpha} + (1-P_{-\beta,\alpha})^2} X_0 U^\dagger \ket{\alpha,-\beta} \nonumber\\
    &= \bra{-\beta,\alpha}U\ket{-\beta,\alpha;0} \bra{\alpha,-\beta;0} U^\dagger \ket{\alpha,-\beta} + \bra{-\beta,\alpha}U(1-P_{-\beta,\alpha}) X_0(1-P_{\alpha,-\beta}) U^\dagger \ket{\alpha,-\beta},
\end{align}
where we have used $1-P_{-\beta,\alpha}=(1-P_{-\beta,\alpha})^2$ for the projector in the first line. Similarly, \begin{align}
    \bra{-\beta,\alpha}Z_1 UX_0 U^\dagger \ket{\alpha,\beta} &= -\bra{\beta,\alpha}U\mlr{P_{\beta,\alpha} + (1-P_{\beta,\alpha})^2} X_0 U^\dagger \ket{\alpha,\beta} \nonumber\\
    &= -\bra{\beta,\alpha}U\ket{\beta,\alpha;0} \bra{\alpha,\beta;0} U^\dagger \ket{\alpha,\beta} - \bra{\beta,\alpha}U(1-P_{\beta,\alpha}) X_0(1-P_{\alpha,\beta}) U^\dagger \ket{\alpha,\beta}.  
\end{align} 
Subtracting the above two equations and choosing $\alpha=1,\beta=0$ so that the two first terms are opposite, we have \begin{align}
    \abs{\bra{0,1}[UX_0 U^\dagger, Z_1 ]\ket{1,0}} &\ge 2\abs{\bra{0,1}U\ket{0,1;0} \bra{1,0;0} U^\dagger \ket{1,0}} - 2 \norm{(1-P_{0,1}) U \ket{0,1}}^2 \nonumber\\
    &\ge 2 (1-\epsilon) - 2\epsilon = 2(1-2\epsilon) = 4\delta,
\end{align}
where we have used the definition of $\epsilon$ in the main text Eq.~(5).
This establishes \eqref{eq:UXUZ>1} and thus \eqref{eq:T>logNN}.
\end{proof}

\section{Further details of the protocol}

\subsection{Semiclassical analysis of the TAT dynamics}
Here we briefly review the semiclassical trajectories of TAT, which is useful to understand our protocol.

Define the normalized angular momentum $X^a\cl:=X^a/N$, which satisfy commutation relation $
    [X\cl, Y\cl] = N^{-2}[X,Y]=N^{-2} 2\ii Z=: 2\ii \hbar Z\cl$,
where $\hbar:=1/N$ is the effective Planck constant
. Then $S_\tau = \ee^{-\ii (\tau\log N) H\cl/\hbar}$ where $H\cl := X\cl Y\cl + Y\cl X\cl \approx 2X\cl Y\cl$
can be viewed as a classical Hamiltonian on the unit-sphere phase space if $\hbar\ll 1$, i.e. $N\gg 1$. In this semiclassical limit, roughly speaking, the initial state $\kzero$ corresponds to an ensemble of initial points near the north pole $Z\cl=1$, and the evolved state $\ket{\Psi}:=S_{\tau} \kzero$ corresponds to the ensemble of these points evolved by the classical trajectories of $H\cl$. It turns out that the north pole $Z\cl=1$ is a saddle point in phase space with Lyapunov exponent $\lambda=4$, so the ensemble is squeezed \emph{exponentially} in the $Y$ 
direction, and stretched exponentially in the $X$ direction. 
Since the initial ensemble has width $\alr{\Delta Y}_{\bm{0}}\sim \sqrt{N}$ (here we define $\alr{\OO}_\psi:=\alr{\psi|\OO|\psi}$, and $\alr{\Delta Y}_\psi:=\sqrt{\alr{Y^2}_\psi-\alr{Y}_\psi^2}$.), the extreme (i.e. asymptotically optimal) squeezing $\alr{\Delta Y}_\Psi \sim 1$ is achieved when $\ee^{-\lambda (\tau\log N)}\times 1/\sqrt{N}\approx 1/N$, which leads to \begin{equation}
    \tau=\tau_{\rm min} \approx 1/8.
\end{equation}

\comment{
): \begin{equation}\label{eq:extremeSS}
    \frac{\alr{\Delta Y}_\Psi}{\alr{Z}_\Psi}=\Theta\lr{\frac{1}{N}},
\end{equation}
where $\ket{\Psi}=S_{\tau_{\rm min}} \kzero$ is the final state,
and $\alr{\Delta Y}_\Psi:=\sqrt{\alr{Y^2}_\Psi-\alr{Y}_\Psi^2}$,  . The ratio \eqref{eq:extremeSS} is much smaller than that $\Theta( 1/\sqrt{N})$ for the initial spin-coherent state $\kzero$. This squeezing dynamics
}

This semiclassical analysis implies that in our protocol, $\tau_2$ is roughly determined by $N,\theta$: Since $S_{\tau_2}$ stretches an initial distance $2N\phi$ to roughly $N$, $\ee^{\lambda(\tau_2 \log N)}\times N\phi \approx cN$, so that \begin{equation}\label{eq:tau2=theta}
    \tau_2\approx \frac{\log(c/\phi)}{4\log N}=\frac{\log(cN/\theta)-2\log(\log N)}{4\log N}.
\end{equation}
Here the constant $c=\Theta(1)$ is a subdominant contribution, and comes from the the fact that the exponential acceleration behavior $\ee^{\lambda \tau \log N}$ gets modified away from the saddle point.

\subsection{Symmetry of the protocol}
Write the protocol $U$ by $U=U_{\rm later} U_{\rm sep}$ where $U_{\rm sep}=C_\phi S_{\tau_1}$, and $U_{\rm later}$ is the later stages in Eq.~(6) in the main text. Since $U_{\rm sep}(\ket{z}_0\otimes\kzero)= \ket{z}_0\otimes U_{\mathrm{sep},z}\kzero$ ($z=0,1$) where the two $z$s are related by a $\pi$-rotation symmetry: $U_{\mathrm{sep},1} = R^Z_{\pi} U_{\mathrm{sep},0}$, we have \begin{align}\label{eq:U=RX+YU}
    U \ket{\alpha,\beta;0} &= \alpha \ket{0}_0\otimes U_{\rm later}U_{\mathrm{sep},0} \kzero + \beta\ket{1}_0\otimes U_{\rm later}R^Z_{\pi} U_{\mathrm{sep},0}\kzero \nonumber\\ &= \alpha \ket{0}_0\otimes U_{\rm later}U_{\mathrm{sep},0} \kzero + \beta\ket{1}_0\otimes R^{X+Y}_{\pi} U_{\rm later} U_{\mathrm{sep},0}\kzero \nonumber\\ &=: \alpha \ket{0}_0\otimes \ket{\psi} + \beta\ket{1}_0\otimes R^{X+Y}_{\pi}\ket{\psi},
\end{align}
where $R^{X+Y}_{\pi}$ is $\pi$-rotation along $\hat{x}+\hat{y}$ direction. In the last step of \eqref{eq:U=RX+YU}, we have used the following commutation relations \begin{align}
    U_{\rm later}{\color{red}R^Z_{\pi}} &= S_{\tau_3} \lr{R^Z_{\pi/4} O_{\pi/4} R^X_{-\pi/2}}{\color{red}R^Z_{\pi}} S_{-\tau_2} \nonumber\\&= S_{\tau_3} \lr{R^Z_{\pi/4} O_{\pi/4}{\color{red}R^Y_{\pi}} R^X_{-\pi/2}} S_{-\tau_2} \nonumber\\
    &= S_{\tau_3} \lr{R^Z_{\pi/4} {\color{red}R^Y_{\pi}} O_{\pi/4} R^X_{-\pi/2}} S_{-\tau_2} \nonumber\\
    &= S_{\tau_3} {\color{red}R^{X+Y}_{\pi}} \lr{R^Z_{\pi/4}  O_{\pi/4} R^X_{-\pi/2}} S_{-\tau_2} \nonumber\\
    &= {\color{red}R^{X+Y}_{\pi}} S_{\tau_3} \lr{R^Z_{\pi/4}  O_{\pi/4} R^X_{-\pi/2}} S_{-\tau_2} = {\color{red}R^{X+Y}_{\pi}} U_{\rm later},
\end{align}
where the third line comes from $\lr{R^Y_\pi}^\dagger=R^Y_\pi$ and $R^Y_\pi Z^2 R^Y_\pi = (R^Y_\pi Z R^Y_\pi)^2= (-Z)^2=Z^2$; the final line is similarly from $H_{\rm TAT} \propto (X+Y)^2-(X-Y)^2$.
Taking inner product between \eqref{eq:U=RX+YU} and the goal $\ket{\alpha,\beta}$, we have \begin{equation}
    \epsilon=1-\min_{\alpha,\beta}\abs{\abs{\alpha}^2 \alr{\bm{0}|\psi} + \abs{\beta}^2 \alr{\bm{1}|R^{X+Y}_\pi|\psi}}^2 = 1-\min_{\alpha,\beta}\abs{\lr{\abs{\alpha}^2 + \abs{\beta}^2}\alr{\bm{0}|\psi}}^2 = 1-\abs{\alr{\bm{0}|\psi}}^2,
\end{equation}
claimed in the main text.

\subsection{Additional numerical data}

\begin{figure}
    \centering
    \includegraphics[width=0.8\textwidth]{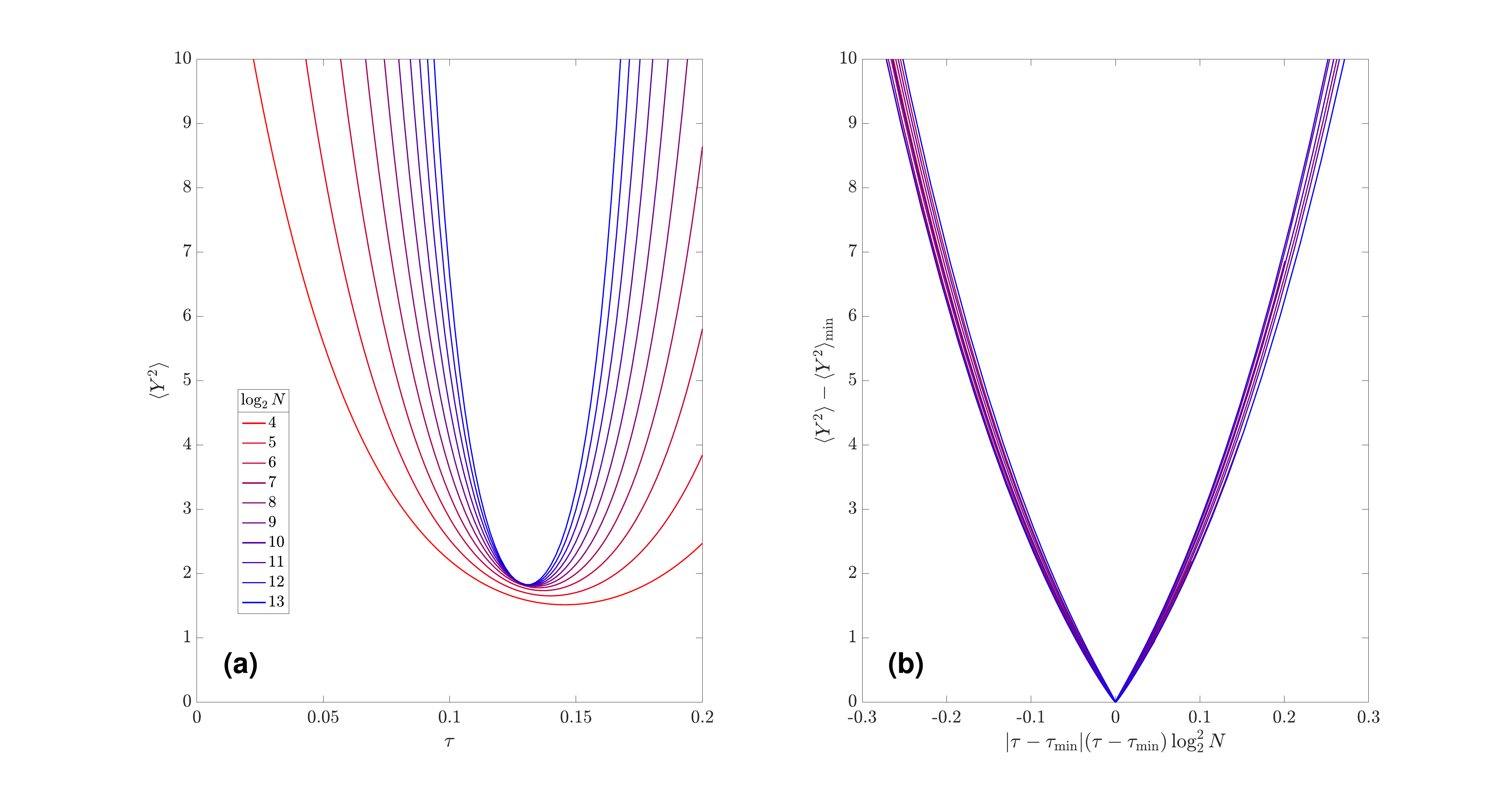}
    \caption{(a) TAT dynamics $S_\tau$ evolves initial state $\kzero$ in rescaled time $\tau$ to a squeezed state with small $\alr{\Delta Y}^2=\alr{Y^2}$ (since $\alr{Y}=0$), which achieves minimum at $\tau=\tau_{\rm min}\approx 1/8$. Here $\tau_{\rm min}$ is numerically determined for each $N$ (each color). The collapse of curves at $\tau \approx 1/8$ indicates the squeezing is extreme $\Delta Y =\Theta(1)$ at that point. (b) The same data but with axes rescaled, where $\alr{Y^2}_{\rm min}$ is the value of $\alr{Y^2}$ at $\tau_{\rm min}$. The collapse of curves shows that close to the extreme point, $\alr{\Delta Y}=\Theta( \log N)$ for any nonvanishing $\tau-\tau_{\rm min}$. This leads to an extra $\log N$ factor in the controlled-rotation angle $\phi$.}
    \label{fig:TAT}
\end{figure}

\begin{figure}
    \centering
    \includegraphics[width=\textwidth]{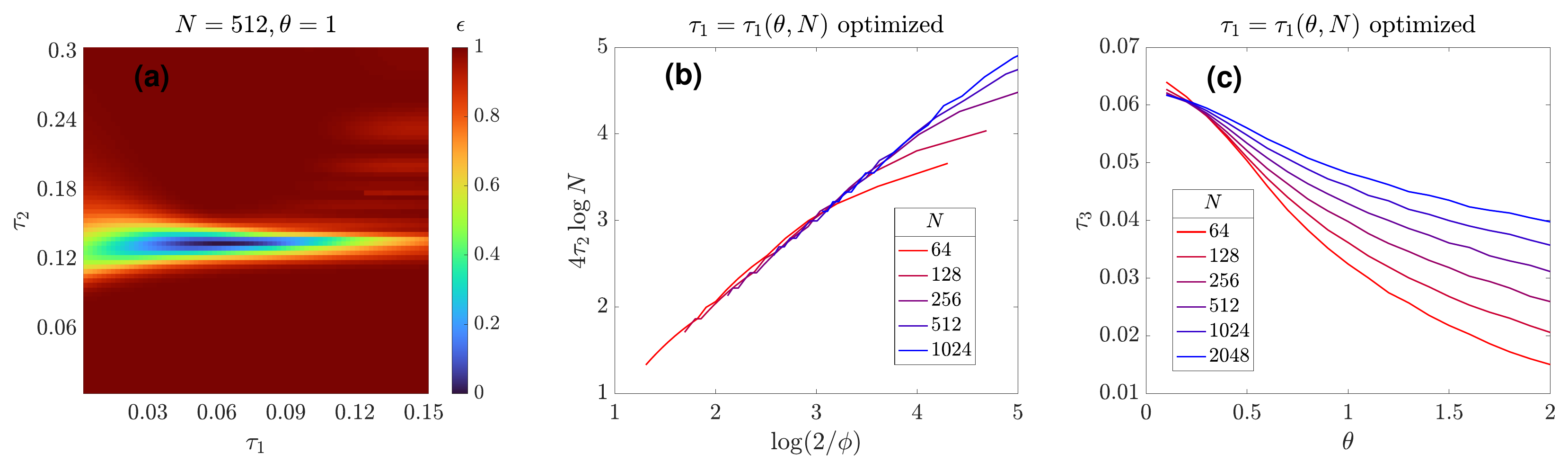}
    \caption{(a) For fixed $N,\theta$ in the figure title, the encoding infidelity is small in a very narrow window of $\tau_2$, which almost does not depend on $\tau_1$. (b) The optimal $\tau_2$ as a function of $N,\theta$, where $\tau_1$ is chosen at $\tau_1(\theta,N)$: the $\tau_1$ value of the optimal pixel in plots like (a). The collapse of curves on the slope-$1$ line verifies prediction \eqref{eq:tau2=theta} for $\tau_2$ with constant $c=2$. (c) The optimized $\tau_3$ values for the duration of the last unsqueezing stage $S_{\tau_3}$, which minimizes the infidelity $\epsilon$ for each set of $(N,\theta,\tau_1,\tau_2)$. Here $\tau_1,\tau_2$ are chosen as the optimal ones for each $(N,\theta)$. }
    \label{fig:tau2tau3}
\end{figure}

In the first stage of our protocol, we squeeze the initial state in order to separate it to two parts in short time. More precisely, the controlled-rotation angle $\phi$ needs to be much larger than the squeezed quantum fluctuation $\Delta Y/N$ (say by a $\log N$ factor). Although $\Delta Y=\Theta(1)$ at extreme squeezing $\tau_{\rm min}\approx 1/8$, we find that beyond (but near) this particular time, the scaling becomes $\Delta Y=\Theta(\log N)$, as shown in Fig.~\ref{fig:TAT}. This is the reason that we choose $\phi\sim \log^2 N/N$, because we do not want to work at extreme squeezing, which is hard to ``unsqueeze''.

In Fig.~\ref{fig:tau2tau3}(a,b), we verify that our protocol performs well at $\tau_2$ around the predicted value \eqref{eq:tau2=theta}, which does not depend on $\tau_1$; the results are similar for other values of $N,\theta$. 

We always numerically optimize parameter $\tau_3$ to maximize the overlap $\abs{\alr{\bm{0}|\psi}}$ in the last unsqueezing stage; the obtained values are shown in Fig.~\ref{fig:tau2tau3}(c), where for a given set of $(N,\theta)$, we focus on $\tau_1,\tau_2$ that minimize the final infidelity. We find that $\tau_3$ decreases with (increasing) $\theta$ and increases with $N$. The reason is that $\tau_3$ is determined by how squeezed the state is after the pulling-away stage $S_{-\tau_2}$: more squeezing requires larger $\tau_3$ to unsqueeze its effect. As a result, $\tau_3$ should grow with $\tau_2$ that is roughly the previous squeezing time. The behavior of $\tau_3$ then agrees with \eqref{eq:tau2=theta}: $\tau_2\approx 1/4-\log(\theta)/\log N$ decreases with $\theta$ and increases with $N$.

\revise{In the main text, we have argued $\alr{Z}_\psi=N-\order(\epsilon)$ with a small relative error $\sim \epsilon/N$ based on the exponential-decaying support of $\ket{\psi}$ on the Dicke states. In Fig.~\ref{fig:Z}, we plot $\alr{Z}_\psi$ directly, and observe that its relative error decays as a power law with $N$ for all system sizes $N\lesssim 2000$ studied. This is in contrast to the GHZ encoding error $\epsilon$, which slightly grows with $N$ at large $N$ (see Fig.~3(b) in the main text). This relative error $\sim \epsilon/N$ of $\alr{Z}_\psi$ makes the error of the enhanced two-qubit gate extremely small. Furthermore, this relative error also quantifies whether a signal from the original qubit $0$ (the data qubit) fully propagates to any other qubit $i>0$, because one can measure $\alr{Z_i}_{\widetilde{\psi}}=\alr{Z}_{\widetilde{\psi}}/N$ in the final state $\widetilde{\psi}$ to tell whether the original qubit was in $\ket{0}$ or $\ket{1}$. If the relative error shown in Fig.~\ref{fig:Z} continues to decay (or saturates at a constant) with $N$, our protocol would saturate (up to logarithmic factors) the Lieb-Robinson-type bound $t=\tilde{\Omega}(N^{\frac{\gamma}{d}-1})$ derived in \cite{Wprotocol_gorshkov20} for signaling using power-law interactions with exponent $\gamma<d$. This is the only regime where the ultimate speed limit is still an open problem in mathematical physics \cite{power_LRB21,ourreview}, where our protocol strongly suggests the solution. }

\begin{figure}
    \centering
    \includegraphics[width=0.4\textwidth]{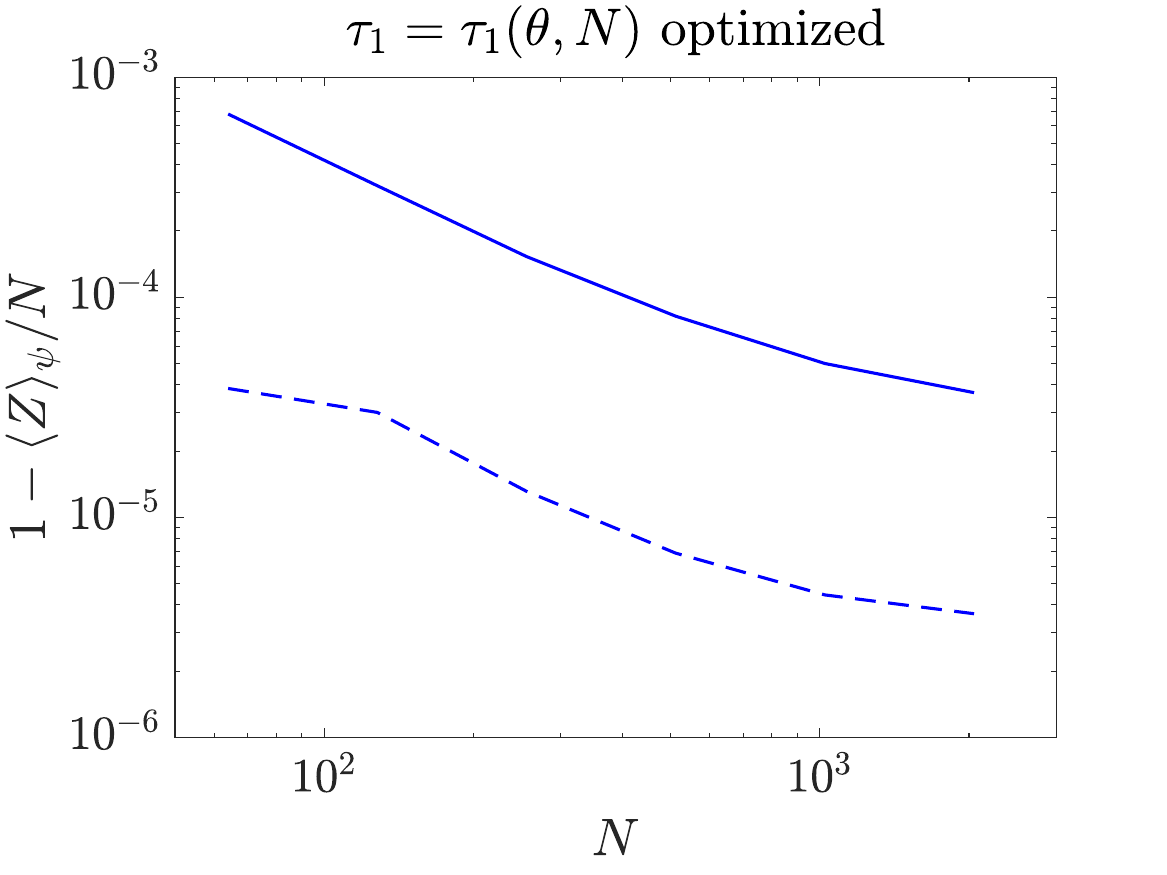}
    \caption{\revise{Relative error of the total $Z$ polarization in the $\alpha$-part final state $\ket{\psi}$ at $\theta=1$ (solid) and $\theta=2$ (dashed), where $\tau_p$s are optimized. The two axes are in log scale.} }
    \label{fig:Z}
\end{figure}

\section{Robustness of our protocol}

\subsection{Faster evolution leads to smaller decoherence error}\label{sec:decoh}

\revise{
Here we explain quantitatively that our fast protocol accumulates less decoherence error comparing to e.g. the $\Theta(1)$-time protocols, a desirable property when using the protocol to encode/decode quantum error correction codes. Let $\rho$ denote the system density matrix, we expect \begin{equation}
    \norm{\delta \rho}_1\sim \Gamma T,
\end{equation} 
for error $\delta \rho$ on the density matrix at weak decoherence, where $\Gamma$ is a global decoherence rate time-averaged over the process. 
Since most of the time our protocol stays in a very squeezed state, which is nearly the most vulnerable state under decoherence, $\Gamma\approx \Gamma_{\rm max}$ is roughly the maximal possible decoherence rate. However, since the target GHZ state also decoheres in the fastest possible way, a similar rate holds for the parallelizing-CNOTs protocol. Therefore, comparing to that $\Theta(1)$-time protocol, our error $\norm{\delta \rho}_1$ is much smaller by a factor of $T=\torder(1/N)$. 
}

\subsection{Robustness of our protocol against inhomogeneous coupling coefficients}

\begin{figure}
    \centering
    \includegraphics[width=0.95\linewidth]{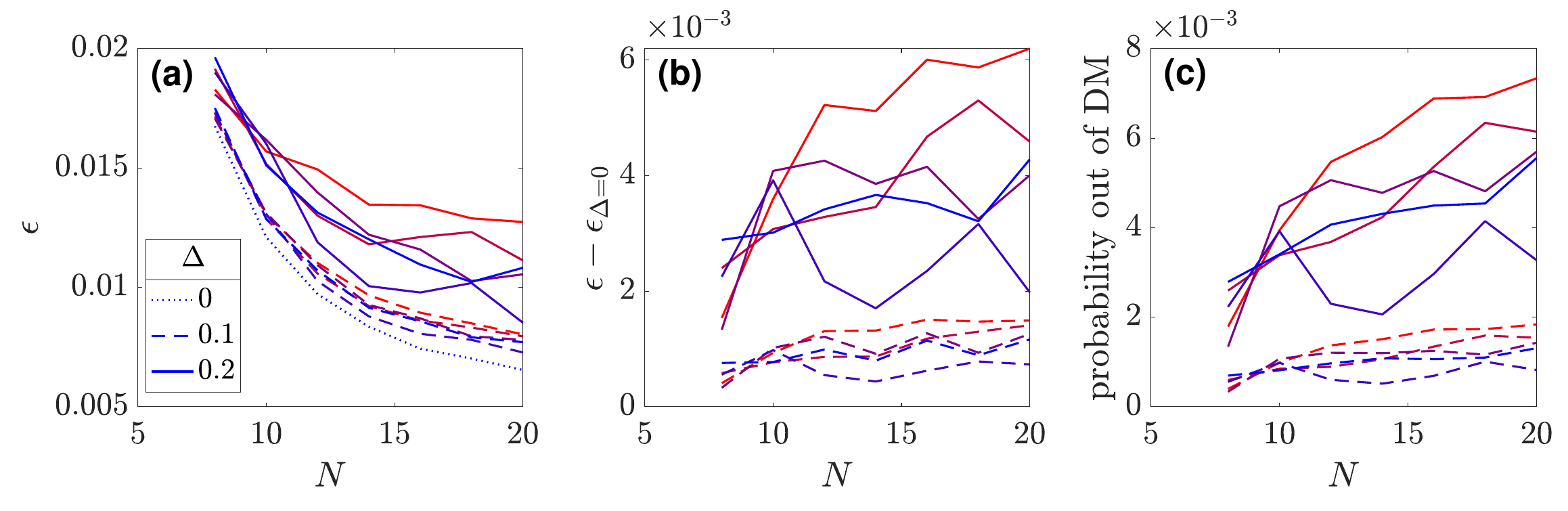}
    \caption{Performance of our protocol where the TAT Hamiltonian is replaced by \eqref{eq:TAT_inho} with inhomogeneous coupling coefficients. We choose $\theta=1$ so that for the range $N\in [8,20]$ studied, the runtime $T$ is shorter than the parallelizing-CNOTs protocol by a factor $\lesssim 1/3$. $\tau_1,\tau_2,\tau_3$ are optimized. The different colors represent five different realizations of the disorder. Note that in the range studied, the infidelity for the ideal case $\epsilon_{\Delta=0}$ decreases with increasing $N$, in contrast to Fig.~1 in the main text at larger $N$.  }
    \label{fig:robust}
\end{figure}

Since our protocol is understood via semiclassical dynamics in the permutation-invariant subspace, the DM of the $N$ qubits, an important question arises: Does our protocol rely crucially on the permutation symmetry? This is related to potential generalizations to e.g. power-law interacting systems. Here we provide some evidence showing certain robustness of our protocol when the permutation symmetry is broken.

We assume that the coupling coefficient is inhomogeneous (disordered) for the TAT subroutines (for simplicity we do not add disorder to the other subroutines): Instead of $H_{\rm TAT}=XY+YX$, the Hamiltonian becomes \begin{equation}\label{eq:TAT_inho}
    H_{\rm TAT}' = \sum_{i}\sum_{j\neq i} J_{i j} (X_iY_j+Y_i X_j).
\end{equation}
For each pair of $i,j$, we independently choose $J_{ij}\in [1-\Delta, 1+\Delta]$ from the uniform distribution in the interval, where $\Delta$ tunes the disorder strength. Since the dynamics is no longer constrained in the DM, we perform numerics by simulating the total Hilbert space for system size up to $N=20$. This system size is already relevant for most current experiments, and the result is reported in Fig.~\ref{fig:robust}.

Naively, one may worry that since DM is only an exponentially small subspace in the total Hilbert space, once the permutation symmetry is broken, the state could easily exit the subspace and never come back. This would lead to failure of our protocol $\epsilon\approx 1$, and the best one can hope for would be a GHZ-like encoding protocol, where the final state is of the form $\alpha \ket{\text{very positively polarized}} + \beta \ket{\text{very negatively polarized}} $. However, this is not what we found in Fig.~\ref{fig:robust}(a) for weak disorder $\Delta$. The infidelity is only slightly increased for $\Delta=0.1$, and remains at the same order of magnitude ($\epsilon\sim 10^{-2}$ for $N=20$) comparing to the ideal case $\epsilon_{\Delta=0}$ at $\Delta=0.2$. The same data is plotted in Fig.~\ref{fig:robust}(b) in terms of the extra infidelity $\epsilon-\epsilon_{\Delta=0}$ due to inhomogeneity, which remains at the order of $10^{-3}$. Remarkably, this is even much smaller than the naive perturbation strength $\Delta$ to the permutation-invariant case. Furthermore, although $\epsilon-\epsilon_{\Delta=0}$ seems to grow with $N$, this effect is mild and looks like $\sim \log N$, which suggests $\epsilon$ may still be small at larger $N\sim 10^2$. 

In Fig.~\ref{fig:robust}(c), we plot the support probability of the state out of the DM during evolution. More precisely, we consider each of the seven time steps between two subroutines in the protocol, and take the maximum of the outside support. Since only the TAT subroutines do not preserve the permutation symmetry, this reduces to only two time steps: after $S_{\tau_1}$ and after $S_{-\tau_2}$. We find that this outside probability is nearly the same (albeit slightly larger) as $\epsilon-\epsilon_{\Delta=0}$ in Fig.~\ref{fig:robust}(b). This indicates that the extra infidelity is almost due to the part of the wave function that is kicked out of the DM, which turns out to be even much smaller than the disorder strength $\Delta$. The robustness of our protocol above thus stems from this robustness of the DM, which is worth a more thorough understanding that we leave as future work. Note that a similar robustness behavior has been studied recently in the different setting of power-law interactions \cite{rotor_spinwave23}.


\end{document}